\documentclass[12pt, onecolumn, letterpaper]{IEEEtran} 
\usepackage[margin=1in]{geometry}
\usepackage{setspace}
\usepackage{amsmath,amssymb,amsthm,bm,amsfonts, mathrsfs,multicol}
\usepackage{algorithmic}
\usepackage{algorithm}
\usepackage{array}
\usepackage[caption=false,font=normalsize,labelfont=sf,textfont=sf]{subfig}
\usepackage{textcomp}
\usepackage{stfloats}
\usepackage{url}
\usepackage{verbatim}
\usepackage{graphicx}
\usepackage{tikz-cd}
\usepackage[numbers]{natbib}
\newtheorem{theorem}{Theorem}
\newtheorem{lemma}[theorem]{Lemma}
\newtheorem{defi}[theorem]{Definition}
\newtheorem*{defi*}{Definition}
\newtheorem{prop}[theorem]{Proposition}
\newtheorem{cor}[theorem]{Corollary}
\newtheorem{example}[theorem]{Example}
\newtheorem{construction}[theorem]{Construction} 
\newcommand{\etal}{\textit{et al.} }

\DeclareMathOperator{\Span}{Span}
\DeclareMathOperator{\ev}{ev}
\DeclareMathOperator{\ord}{ord}
\usepackage{hyperref}
\begin{document}
\title{Quantum Locally Recoverable Codes \\ via Good Polynomials}

	\author{ Sandeep Sharma, Vinayak Ramkumar,  and Itzhak Tamo\\
		\thanks{
Sandeep Sharma, Vinayak Ramkumar,   and Itzhak Tamo are with the Department of Electrical Engineering--Systems, Tel Aviv University, Tel Aviv, Israel. e-mail:	sk.deepsharma@gmail.com,   vinram93@gmail.com,  tamo@tauex.tau.ac.il. \\
	This work was supported by the European Research Council (ERC grant number 852953).}
	}
\maketitle
\begin{abstract}
Locally recoverable codes (LRCs) with locality parameter $r$ can recover any erased code symbol by accessing $r$ other code symbols. This local recovery property is of great interest in large-scale distributed classical data storage systems as it leads to efficient repair of failed nodes. 
A well-known class of optimal (classical) LRCs are subcodes of Reed-Solomon codes constructed using a special type of polynomials called good polynomials. 
Recently, Golowich and Guruswami initiated the study of quantum LRCs (qLRCs), which could have applications in quantum data storage systems of the future. The authors presented a qLRC construction based on good polynomials arising out of subgroups of the multiplicative group of finite fields. 
In this paper, we present a qLRC construction method that can employ any good polynomial.  
We also propose a new approach for designing good polynomials using subgroups of affine general linear groups.  Golowich and Guruswami also derived a lower bound on the minimum distance of their qLRC under the restriction that $r+1$ is prime. Using similar techniques in conjunction with the expander mixing lemma, we develop minimum distance lower bounds for our qLRCs without the $r+1$ prime restriction.
\end{abstract}
\section{Introduction}
In large-scale classical data storage systems, the failure of a single storage node is a more common occurrence than the catastrophic failure of multiple nodes simultaneously. This observation led to the development of various classes of codes that can repair single-node failure efficiently. Locally recoverable codes (LRCs) are an important family of such erasure codes designed for distributed storage. The goal of LRCs is to minimize the number of nodes that are contacted to repair a failed node. In an LRC with locality $r$, for any code symbol, there exist at most $r$ other symbols from which the value of the code symbol can be determined. 
Some of the early works that investigated codes having local recovery property are \cite{HanL07, HuangCL07, OggierD11, Gopalan, PrakashKLK12, PapailiopoulosD14, SilbersteinRKV13, MazumdarCW14, Tamo, CadambeM15, TamoPD16, TamoBF16}. 

In \cite{Golowich}, Golowich and Guruswami introduced the quantum counterpart of LRCs, which will be referred to as quantum LRCs (qLRCs). Although quantum computation and storage are still in their infancy, it is possible that large-scale quantum data storage could become a reality in the future, thereby making quantum codes like qLRCs inevitable.
The authors of \cite{Golowich} also pointed to the connections between quantum LDPC codes and qLRCs as an additional motivation to study qLRCs. 
 The CSS construction \cite{calderbank1996good, steane1996error, Ketkar} is a well-known method to construct quantum codes from classical codes. As shown in \cite{Golowich}, using CSS framework, one can reduce the problem of constructing qLRCs to that of constructing classical LRCs which contain their dual. 
 
Of particular interest for our paper is the LRC construction in \cite{Tamo}. This work presented a family of LRCs that are subcodes of Reed-Solomon codes. These LRCs attain the optimal value of the minimum distance and require only an alphabet of size comparable to the block length. The qLRC construction in \cite{Golowich} essentially involves modifying a special case of the LRC construction in \cite{Tamo} to make it dual-containing.

Our work aims to generalize the qLRC construction in \cite{Golowich} to leverage the full potential of the LRC construction in \cite{Tamo}. In detail, the LRC in \cite{Tamo} is obtained by evaluating specially constructed polynomials at a chosen set of points in a finite field. This construction relies on the existence of a special type of polynomials associated with evaluation points, called good polynomials. The annihilator polynomials of cosets of a subgroup of the multiplicative group of the field form one class of good polynomials. In \cite{Golowich}, good polynomials of this particular type are used to construct dual-containing LRCs. Many other ways to obtain good polynomials are known in the literature \cite{Tamo, LiuMC18, Micheli, LiuMT20, RuiKai}. For example, annihilator polynomials of cosets of an additive subgroup of the field also yield good polynomials. In this paper, we present a method to construct dual-containing LRCs that can employ any good polynomial. The flexibility to choose good polynomials enables us to construct qLRCs for parameters that are not achievable using the construction in \cite{Golowich}.

Additionally, we propose a general approach for constructing good polynomials, which, to the best of our knowledge, is not previously known in the literature. This approach is based on subgroups of the affine general linear group and encompasses the good polynomial constructions in \cite{Tamo} as special cases. We also present a lower bound on the minimum distance of our qLRCs constructed using good polynomials based on this affine general linear group approach. This lower bound can be viewed as a generalization of the bound given in \cite{Golowich}.

The rest of the paper is organized as follows. In Section~\ref{sec:prelim}, we introduce relevant notions and provide definitions of LRC and qLRC. Our qLRC construction is presented in Section~\ref{sec:constrcution}. In Section~\ref{sec:AGL}, we present a method to obtain good polynomials using subgroups of the affine general linear group. In Section~\ref{sec:distance_bound}, we derive lower bounds on the minimum distance of our qLRCs. We conclude the paper in Section~\ref{sec:conclusion}.

\section{Preliminaries}\label{sec:prelim}
A linear code \( \mathcal{C} \) of block length \( n \) and dimension \( k \) over \( \mathbb{F}_{q} \), the finite field of order \( q \), is a \( k \)-dimensional subspace of \( \mathbb{F}_{q}^n \). The minimum distance is defined as \( d = \min\{\text{wt}(\boldsymbol{c}) : \boldsymbol{c} \in \mathcal{C}, \boldsymbol{c} \neq \mathbf{0}\} \), where \( \text{wt}(\boldsymbol{c}) \) denotes the Hamming weight of \( \boldsymbol{c} \).
The support of a vector \( \boldsymbol{v} \in \mathbb{F}_{q}^n \), denoted by \( \text{supp}(\boldsymbol{v}) \), is the set of positions where \( \boldsymbol{v} \) has nonzero coordinates.

The dual code \( \mathcal{C}^\perp \) of a linear code \( \mathcal{C} \) is defined as
$$
\mathcal{C}^\perp = \left\{ (x_1, x_2, \ldots, x_n) \in \mathbb{F}_q^n : \sum_{i=1}^{n} c_i x_i = 0 \text{ for all } (c_1, c_2, \ldots, c_n) \in \mathcal{C} \right\}.
$$
A linear code \( \mathcal{C} \) is said to be dual-containing if it satisfies \( \mathcal{C}^\perp \subseteq \mathcal{C} \), which holds only if \( k \geq \frac{n}{2} \). We denote \( [n] \) as the set \( \{1, 2, \ldots, n\} \). Bold symbols like \( \boldsymbol{u} \) are used to denote vectors.

\subsection{Classical Locally Recoverable Codes}
In this subsection, we provide a brief overview of LRCs.  We restrict our attention to linear LRCs. 
\begin{defi}[Locally Recoverable Code]
Let $\mathcal{C}$ be an $[n,k]_q$ linear 
code over $\mathbb{F}_q$, and let $r$ be a positive integer. We say that $\mathcal{C}$ has locality $r$ if each symbol $c_i$ of a codeword $\boldsymbol{c} = (c_1, c_2, \ldots, c_n) \in \mathcal{C}$ can be recovered from at most $r$ other symbols (different from $c_i$) in $\boldsymbol{c}.$ That is, for each $i \in [n]$, there exists a subset $I_i \subseteq [n] \setminus \{i\}$ with $|I_i| \leq r$ such that $c_i$ can be recovered from the coordinates indexed by $I_i.$ The set $I_i$ is called a recovery set for $c_i$.
\end{defi}
Gopalan \etal \cite{Gopalan} established that every $[n,k,d]_q$ linear code with locality $r$ satisfies a Singleton-like bound:
\begin{equation*}\label{Singleton}
    d \leq n - k - \left\lceil \frac{k}{r} \right\rceil + 2.
\end{equation*}
LRCs that achieve this bound with equality are referred to as optimal LRCs and a family of optimal LRCs is given in \cite{Tamo}.
We refer readers to the survey in \cite{survey} for a detailed discussion on LRCs and their variants.

\subsection{Quantum Locally Recoverable codes}
We begin this subsection by defining quantum codes. 
\begin{defi}[Quantum Code]
A quantum code of block length $n$, dimension $\kappa$, and alphabet size $q$ is a $q^{\kappa}$-dimensional subspace $\mathcal{C}$ of the Hilbert space $(\mathbb{C}^q)^{\otimes n}$. 
\end{defi}
A quantum code is typically characterized by its parameters $[[n,\kappa, \delta]]_q$, where $\delta$ is the minimum distance of the code. Such a quantum code encodes $\kappa$ logical qudits into entangled states of $n$ qudits
and protects against the erasure of any set of $\delta-1$ qudits. 
In quantum mechanics, the Dirac notion is commonly used, under which codewords of quantum code $\mathcal{C}$ are denoted as $\vert \phi \rangle \in \mathcal{C}$.

Now, we shall recall a special type of quantum codes: CSS codes, named after Calderbank, Shor, and Steane. CSS construction \cite{calderbank1996good, steane1996error, Ketkar} is a method of obtaining quantum codes from two classical linear codes that satisfy a specific relationship with one another. Quantum codes constructed through the CSS approach are called CSS codes. 

\begin{prop}[CSS Construction \cite{Ketkar}] Let $\mathcal{C}_i$ be an $[n,k_i]_q$ linear code for $i=1,2.$ If $\mathcal{C}_1^{\perp}\subseteq \mathcal{C}_2$, then there exists an $[[n,\kappa,\delta]]_q$-quantum code $\mathcal{C}=\text{CSS}(\mathcal{C}_1,\mathcal{C}_2)$ with $\kappa=k_1+k_2-n$ and $\delta=\min\{\text{wt}(\mathcal{C}_1\setminus\mathcal{C}_2^{\perp}),\text{wt}(\mathcal{C}_2\setminus\mathcal{C}_1^{\perp})\}.$ 
\end{prop}
The quantum version of LRC introduced by \cite{Golowich} can be defined as follows.
\begin{defi}[Quantum Locally Recoverable Code \cite{Golowich}] An $[[n,\kappa, \delta]]_q$ quantum code has locality $r$ if any single qudit of $\vert \phi \rangle \in  \mathcal{C}$ is erased then $\vert \phi \rangle$ can be recovered by applying a recovery channel that accesses only $r$ other qudits of $\vert \phi \rangle$. 
\end{defi}
Proposition \ref{prop:lrc_css} below gives the CSS approach to construct qLRCs from classical codes. 
\begin{prop} 
[\cite{Golowich}] \label{prop:lrc_css}
    A quantum code $\mathcal{C}=\text{CSS}(\mathcal{C}_1,\mathcal{C}_2)$ has locality $r$ if for each position $i\in [n]$, there exist codewords $ \boldsymbol{c}_1\in \mathcal{C}_1^{\perp}$ and $ \boldsymbol{c}_2\in\mathcal{C}_2^{\perp}$ such that $i\in \text{supp}( \boldsymbol{c}_1)\cap\text{supp}( \boldsymbol{c}_2)$ and $|\text{supp}( \boldsymbol{c}_1)\cup\text{supp}( \boldsymbol{c}_2)|\leq r+1.$
\end{prop}
The following corollary of proposition~\ref{prop:lrc_css} reduces the problem of constructing qLRCs to that of constructing dual-containing classical LRCs.  
\begin{cor} \label{cor1}
 Let $\mathcal{C}$ be an $[n,k]_q$ linear code such that $\mathcal{C}^{\perp}\subset\mathcal{C}.$ If $\mathcal{C}$ has locality $r$, then  $\text{CSS}(\mathcal{C},\mathcal{C})$ is an $[[n,2k-n,\delta]]_q$ quantum code with locality $r$, where $\delta=\text{wt}(\mathcal{C}\setminus\mathcal{C}^\perp).$
\end{cor}

\subsection{Prior Work on qLRCs}
An early work investigating dual-containing LRCs is \cite{MukhopadhyayHB22}, which presents some examples and basic properties of such codes. The authors mention the possibility of constructing quantum codes from these LRCs, but this direction was not further explored.

Motivated by potential applications to quantum storage, qLRCs were formally defined in \cite{Golowich}. The authors of \cite{Golowich} showed that the parameters \( n, \kappa, \delta, r \) of qLRCs must satisfy the following quantum Singleton-like bound:
\begin{equation} \label{eq:quantum_singleton}
  \kappa \le n - 2(\delta - 1) - \Big\lfloor \frac{n - (\delta - 1)}{r + 1} \Big\rfloor - \Bigg\lfloor \frac{n - 2(\delta - 1) - \Big\lfloor \frac{n - (\delta - 1)}{r + 1} \Big\rfloor}{r + 1} \Bigg\rfloor.   
\end{equation}
In \cite{Golowich}, a qLRC construction is provided for parameters satisfying \( n = q - 1 \) and \( (r + 1) \mid n \). This construction follows the CSS approach given by Corollary~\ref{cor1}, and the required dual-containing LRC is obtained by modifying a special case of the LRC in \cite{Tamo}. A folded version of this qLRC is also presented in \cite{Golowich}, which is shown to approach the quantum Singleton-like bound \eqref{eq:quantum_singleton} for large values of \( r \). Additionally, a random construction of qLRCs is presented in \cite{Golowich}, achieving a slightly better tradeoff than the folded construction.

In \cite{luo2023bounds}, using a parity-check matrix viewpoint \cite{Jin19, ChenH19e, XingY22, Zhang20a}, two classical codes satisfying the requirements of Proposition~\ref{prop:lrc_css} are provided for certain parameters. It is also shown that the cyclic LRCs in \cite{TamoBGC15, LuoXY19} are dual-containing for some parameters. The constructions in \cite{luo2023bounds} result in qLRCs with small minimum distance; more specifically, \( 2\delta < r - \frac{n}{r + 1} + 4 \). However, these codes meet the quantum Singleton-like bound~\eqref{eq:quantum_singleton} with equality whenever they exist.

\section{Constructions} \label{sec:constrcution}
In this section, we shall construct qLRCs of length $n$ and locality $r$ over a finite field of size $q\geq n$, where we assume throughout that $(r+1)$ divides $ n$. 

In \cite{Tamo}, the third author and Barg constructed classical LRCs whose codewords are the evaluations of specially constructed polynomials, and these codes can be viewed as subcodes of Reed-Solomon codes. However, these codes are generally not dual-containing, which means that qLRCs cannot be directly constructed from them using the CSS construction. Golowich and Guruswami \cite{Golowich} modified a special case of the construction from \cite{Tamo} and constructed dual-containing LRCs as a subcode of the Reed-Solomon code.

We follow their approach, and in particular, we construct dual-containing LRCs that can be viewed as a subcode of a generalized Reed-Solomon (GRS) code, applying the CSS construction to obtain qLRCs. Hence, the codes we construct are similar to the classical LRCs, but they are defined with some column multipliers (as in GRS) to guarantee dual-containment. 

Lastly, our codes can be viewed as the quantum analog of the classical LRC construction in \cite{Tamo}, 
as any partition of the set of evaluation points together with a good polynomial that is used to construct an LRC in \cite{Tamo} can also be used to construct a qLRC using our construction.

Let $\mathbb{F}_q[x]$ denote the ring of polynomials in $x$ with coefficients from $\mathbb{F}_q.$ For a non-empty subset $S \subset \mathbb{F}_q[x]$, the span of $S$ over $\mathbb{F}_q$, denoted by $\Span(S)$, is defined as
$$\Span(S) = \left\{ \sum\limits_{i=1}^{k} w_i f_i(x) \mid f_i(x) \in S, \, w_i \in \mathbb{F}_q, \, k \in \mathbb{N} \right\}.$$ 
Let $n\leq q$ and $A=\{\alpha_1,\alpha_2,\ldots,\alpha_n\}\subseteq \mathbb{F}_q$ be a set of size $n$. 
Let $\boldsymbol{u}=(u_1,u_2,\ldots,u_n)\in(\mathbb{F}_q^\ast)^n.$ Define the evaluation map $\ev_{A,\boldsymbol{u}}:\Span(S)\rightarrow \mathbb{F}_q^n$ by $$\ev_{A,\boldsymbol{u}}(f(x))=(u_1f(\alpha_1),u_2f(\alpha_2),\ldots,u_nf(\alpha_n))$$ for all $f(x)\in\Span(S)$.

A key ingredient for the LRC construction in \cite{Tamo} is good polynomials, and they play a similarly crucial role in the dual-containing LRCs that we construct.
\begin{defi}[\cite{Tamo}] Let $A\subseteq\mathbb{F}_q$ be a set of size $n$ and let $r$ be a positive integer such that $(r+1)\mid n.$ Let $\{A_1,A_2,\ldots,A_\frac{n}{r+1}\}$ be a partition of $A$ such that $|A_i|=r+1$ for $1\leq i\leq \frac{n}{r+1}.$ A polynomial $g(x)\in\mathbb{F}_q[x]$ of degree $r+1$ is called a \emph{good polynomial} for the $A_i$'s if it is constant on each one of them, i.e., for any $1\leq i\leq \frac{n}{r+1}$ and any $\alpha,\beta\in A_i$, we have $g(\alpha)=g(\beta)$. 
\end{defi}
Next, we describe the qLRC construction. 
\begin{construction} \label{constr}
Let $A= \{\alpha_1,\alpha_2,\ldots,\alpha_n\}=\sqcup_{i=1}^{\frac{n}{r+1}} A_i \subseteq \mathbb{F}_q$ be a set of size $n\leq q$, where $|A_i|=r+1$ for each $i$ and  $g(x)$ is  a good polynomial for the $A_i$'s. Suppose there exists 
 a nonzero vector $\boldsymbol{u}=(u_1,u_2,\ldots,u_n) \in \mathbb{F}_q^n$ such that 
\begin{equation}\label{tiki}
\begin{pmatrix}
1 &  1& \cdots & 1 \\
\alpha_1 &\alpha_2 & \cdots & \alpha_n \\
\vdots & \vdots & \ddots & \vdots \\
\alpha_1^{n-2} &\alpha_2^{n-2} & \cdots & \alpha_n^{n-2}
\end{pmatrix}\begin{pmatrix}
    u_1^2\\
    u_2^2\\
    \vdots\\
    u_n^2
\end{pmatrix}=\bold{0}\footnote{It is easy to verify that such a vector \( \boldsymbol{u} \) always exists over the field \( \mathbb{F}_{q^2} \), and thus the code will be defined over \( \mathbb{F}_{q^2} \). Note that we show later that in many cases it also exists over \( \mathbb{F}_q \), and then the code will be defined over \( \mathbb{F}_{q} \). Furthermore, any such vector \( \boldsymbol{u} \) necessarily satisfies that each \( u_i \) is nonzero.}.
\end{equation}
Define the set $S = S_1 \sqcup S_2$ of $k$ polynomials, where  $\frac{n}{2}<k\leq \frac{nr}{r+1}$, as follows. 
$S_1$ is the set of $k-\frac{n}{r+1}$ polynomials of the form $x^i g(x)^j$ with the smallest possible degree, where $ 1\leq i\leq  r-1$, and 
$S_2$ is the following set of $\frac{n}{r+1}$ polynomials
\begin{equation*} 
    S_2 = \left\{ g(x)^i : i = 0, 1, \ldots, \frac{n}{r+1}-1 \right\}.
\end{equation*}
Then, define the code  \begin{equation*}\label{LRC}
    \mathcal{C}=\left\{\ev_{A,\boldsymbol{u}}(f(x)): f(x)\in \Span(S)\right\}.
\end{equation*} 
\end{construction}
In the sequel, for ease of notation, we often write $\ev(f(x))$ instead of $\ev_{A,\boldsymbol{u}}(f(x))$.

To use the classical code given by Construction~\ref{constr} to obtain a CSS code, we need to show that it is dual-containing. 
The following lemma, whose proof is given in the Appendix~\ref{appendix:good_poly}, is useful in proving this. 
\begin{lemma}\label{lem:good_poly}
Let $g(x)$ be a good polynomial for the partition
  $\mathcal{A}=\{A_1,A_2,\ldots,A_{\frac{n}{r+1}}\}$, where each $A_i$ is of size $r+1$, and let $h(x)$ be the annihilator polynomial of $\sqcup_i A_i$.
Define   $R\subset\mathbb{F}_q[x]$ to be the set of polynomials of degree less than $n$ that are constant on each  $A_i$. Then, $R$ is a vector space of dimension $\frac{n}{r+1}$ with a basis  $\{1,g(x),g(x)^2,\ldots,g(x)^{\frac{n}{r+1}-1}\}$  
and it also forms a ring under multiplication modulo $h(x).$
\end{lemma}
The following theorem shows that the CSS code obtained from the classical code given by Construction~\ref{constr} is a qLRC. 
\begin{theorem}\label{th1}
The linear code $\mathcal{C}$ defined by Construction~\ref{constr} is an $[n,k]_q$ code with locality $r$. Moreover, it is dual-containing. Therefore, the CSS code $\mathscr{C}=\text{CSS}(\mathcal{C},\mathcal{C})$ is a qLRC of length $n$, dimension $2k-n$, alphabet size $q$ and locality $r.$ \end{theorem}
\begin{proof} 
To prove that the code $\mathscr{C}$ is a qLRC of length $n$ and locality $r$ over $\mathbb{F}_q$, it suffices, by Corollary \ref{cor1}, to show that the code $\mathcal{C}$ is a classical LRC of length $n$ and locality $r$ satisfying $\mathcal{C}^\perp \subset \mathcal{C}$. Clearly, the code $\mathcal{C}$ is a linear code of length $n$ over $\mathbb{F}_q$.

Next, we show that the locality is $r$. Consider an arbitrary polynomial $f(x) \in \Span(S)$ and let $\boldsymbol{c} = (c_1,c_2, \ldots, c_n)$ be the codeword corresponding to it, i.e., $c_i = u_i f(\alpha_i)$ for all $i \in [n]$. Assume that the erased code symbol is $c_z = u_z f(\alpha_z)$, where $\alpha_z \in A_j$. Since $g(x)$ is a good polynomial, by the structure of the polynomials in $\Span(S)$ it holds that the polynomial 
$$
\lambda(x) := f(x) \mod (g(x) - g(A_j))
$$ 
is of degree at most $r - 1$, where $g(A_j)$ is the value that $g(x)$ attains on the set $A_j$. Furthermore, $\lambda(\alpha) = f(\alpha)$ for any $\alpha \in A_j$. Hence, $\lambda(x)$ can be interpolated from the $r$ symbols $\lambda(\alpha_i) = f(\alpha_i) = u_i^{-1} c_i$, $\alpha_i \in A_j \setminus \{\alpha_z\}$, and then we output the value $u_z \lambda(\alpha_z) = u_z f(\alpha_z) = c_z$.

Now, we focus on proving that $\mathcal{C}$ has dimension $k$ and is dual-containing. Let $u := \frac{n}{r + 1}$. Observe that the largest degree in $S_1$ corresponds to the following values of $i$ and $j$:
\begin{itemize}
    \item $i = r - 1$ and $j = \left(\frac{k - u}{r - 1}\right) - 1$ if $(k - u) \equiv 0 \text{ (mod } r - 1)$, and
    \item $i = (k - u) \text{ (mod } r - 1)$ and $j = \left\lfloor \frac{k - u}{r - 1} \right\rfloor$ if $(k - u) \not\equiv 0 \text{ (mod } r - 1)$.
\end{itemize}
Thus, the largest degree in $S_1$ is given by
\begin{equation} \label{eq:largest_deg}
    \ell =
\begin{cases} 
    (r + 1) \left( \frac{k - u}{r - 1} \right) - 2 & \text{if } (k - u) \equiv 0 \text{ (mod } r - 1), \\
    (r + 1) \left\lfloor \frac{k - u}{r - 1} \right\rfloor + (k - u) \text{ (mod } r - 1) & \text{if } (k - u) \not\equiv 0 \text{ (mod } r - 1).
\end{cases}
\end{equation}

One can verify that since $k \leq \frac{nr}{r + 1}$, any polynomial in $S_1$ is of degree at most $n - 2$, and thus so is any polynomial in $\Span(S)$. This implies that the mapping $\ev_{A, \boldsymbol{u}} : \Span(S) \rightarrow \mathbb{F}_q^n$ is injective, and the minimum distance of the code is at least $2$. Furthermore, since $S_1$ and $S_2$ are disjoint, the set $S$ consists of $k$ polynomials of different degrees, and therefore, $\dim(\Span(S)) = k$. Together with the fact that the mapping $\ev_{A, \boldsymbol{u}}$ is injective, we conclude that the code has dimension $k$.

Define the set $T = T_1 \sqcup S_2$ of $n - k$ polynomials, where $T_1$ is the set of $n - k - u$ polynomials of the form $x^i g(x)^j$ with the smallest possible degree and $1 \leq i \leq r - 1$. As before, one can verify that the largest degree of a polynomial in $T_1$ corresponds to the following values of $i$ and $j$:
\begin{itemize}
    \item $i = r - 1$ and $j = \left( \frac{n - k - u}{r - 1} \right) - 1$ if $(n - k - u) \equiv 0 \text{ (mod } r - 1)$, and
    \item $i = (n - k - u) \text{ (mod } r - 1)$ and $j = \left\lfloor \frac{n - k - u}{r - 1} \right\rfloor$ if $(n - k - u) \not\equiv 0 \text{ (mod } r - 1).$
\end{itemize}
Thus, the largest degree in $T_1$ is given by
$$
\ell' =
\begin{cases} 
    (r + 1) \left( \frac{n - k - u}{r - 1} \right) - 2 & \text{if } (n - k - u) \equiv 0 \text{ (mod } r - 1), \\
    (r + 1) \left\lfloor \frac{n - k - u}{r - 1} \right\rfloor + (n - k - u) \text{ (mod } r - 1) & \text{if } (n - k - u) \not\equiv 0 \text{ (mod } r - 1).
\end{cases}
$$

Recall that $S_1$ and $T_1$ contain polynomials of the same form, but $|S_1| = k - u > n - k - u = |T_1|$, where the inequality follows since $k > n / 2$. Therefore, $T_1 \subseteq S_1$, which implies that $T \subset S$ and $\Span(T) \subset \Span(S)$.

Now, let us define the code 
$$
\mathcal{D} = \{\ev_{A, \boldsymbol{u}}(f(x)) : f(x) \in \Span(T)\}.
$$ 
Clearly, $\mathcal{D}$ is a linear code of length $n$, dimension $n - k$, and satisfies $\mathcal{D} \subset \mathcal{C}$. We will now show that $\mathcal{C}^\perp = \mathcal{D}$.

To show that $\mathcal{C}^\perp = \mathcal{D}$, it is enough to show that 
$$
\ev(f_1(x)) \cdot \ev(f_2(x)) = \sum_{i=1}^{n} u_i^2 f_1(\alpha_i) f_2(\alpha_i) = 0
$$
for any $f_1(x) \in S$ and $f_2(x) \in T$. Let $h(x)$ be the annihilator polynomial of $A$, define $z(x) := (f_1(x) f_2(x)) \mod h(x)$, and note that 
$$
\sum_{i=1}^{n} u_i^2 f_1(\alpha_i) f_2(\alpha_i) = \sum_{i=1}^n u_i^2 z(\alpha_i).
$$
Recall that  by \eqref{tiki}, it holds that 
$$
\sum_{i=1}^{n} u_i^2 \alpha_i^j = 0
$$
for $0 \le j \leq n - 2$. Therefore, the result will follow if we show that $\deg(z(x)) \leq n - 2$. For this, we consider the following two cases.

\vspace{0.2cm}
\paragraph*{\bf{Case 1}} Suppose $f_1(x) = x^{i_1} g(x)^{j_1} \in S_1$ and $f_2(x) = x^{i_2} g(x)^{j_2} \in T_1$. In this case, we will need the following lemma, whose proof appears in Appendix \ref{appendix-calc}.
\begin{lemma} \label{calc}
   $\ell$ and $\ell'$ satisfy $\ell + \ell' \leq n - 2$.
\end{lemma}
By Lemma \ref{calc}, $z(x) = f_1(x) f_2(x)$ and is of degree at most $n - 2$.

\vspace{0.2cm}
\paragraph*{\bf{Case 2}} Suppose $f_1(x) = x^{i_1} g(x)^{j_1} \not\in S_1$ or $f_2(x) = x^{i_2} g(x)^{j_2} \not\in T_1$, or both. In this case, $f_1(x) f_2(x) = x^{i_1 + i_2} g(x)^{j_1 + j_2}$, where $i_1 + i_2 < r$. Note that $g(x)^{j_1 + j_2} \mod h(x) \in R$, where $R$ is as defined in Lemma~\ref{lem:good_poly}. Therefore, $g(x)^{j_1 + j_2} \mod h(x)$ is spanned by the polynomials $g(x)^i$ for $i = 0, 1, \ldots, u - 1$, and in particular its degree is at most $(r + 1)(u - 1) = n - (r + 1)$. Therefore, 
$$
z(x) = (f_1(x) f_2(x)) \mod h(x) = x^{i_1 + i_2}(g(x)^{j_1 + j_2} \mod h(x)),
$$
which implies that 
$$
\deg(z(x)) \leq r - 1 + n - (r + 1) = n - 2.
$$

Thus, the code $\mathcal{C}$ is a dual-containing code. From Corollary \ref{cor1}, it follows that $\mathscr{C} = \text{CSS}(\mathcal{C}, \mathcal{C})$ is a quantum code of length $n$, dimension $2k - n$, and locality $r$.
\end{proof}

\vspace{0.2cm}
\paragraph*{\bf{Length of the code}}
For our qLRC construction, we need a set of \( n \) evaluation points \( A \subseteq \mathbb{F}_q \) and a vector \( \boldsymbol{u} \in \mathbb{F}_q^n \) satisfying \eqref{tiki}, i.e., a solution where all entries are quadratic residues. This raises the question of how long a qLRC can be given a finite field.

If the characteristic of the field is two, then every element of \( \mathbb{F}_q \) is a quadratic residue, and therefore such a vector \( \boldsymbol{u} \) always exists, which implies that the constraint in \eqref{tiki} can always be satisfied.

If \( q \) is an odd prime power, there are several ways to construct relatively long codes. For example, by letting \( A = \mathbb{F}_q \) (i.e., a code of length \( q \)), one can easily verify that setting \( u_i = 1 \) for all \( i \) satisfies \eqref{tiki}.
Another example is to let \( A = \{\alpha_1,\alpha_2, \ldots, \alpha_{\frac{q-1}{2}}\} \) be the set of quadratic residues of \( \mathbb{F}_q \). Then, for \( u_i \in \mathbb{F}_q \) satisfying \( u_i^2 = \alpha_i \), one can verify that \eqref{tiki} is also satisfied. Note that for both of these examples, one needs to construct the corresponding partition of \( A \) with the corresponding good polynomial. We note that in both cases, this can be done by relying on additive and multiplicative subgroups of the field. See Section \ref{sec:AGL} for more details.

\vspace{0.2cm}
\paragraph*{\bf{Comparison with the qLRC construction in \cite{Golowich}}}
In \cite{Golowich}, a qLRC construction is presented for parameters such that \( n = q - 1 \) and \( (r + 1) \mid n \). First, a dual-containing LRC is constructed, and then the CSS approach is followed to obtain a qLRC. Specifically, the set of evaluation points is chosen to be \( A = \mathbb{F}_q^* \), where the partition of \( A \) is given by the cosets of the multiplicative subgroup of \( \mathbb{F}_q^* \) of order \( r + 1 \) together with the good polynomial \( g(x) = x^{r+1} \). Using these choices, the authors of \cite{Golowich} constructed a dual-containing LRC by modifying the construction in \cite{Tamo} slightly differently from our approach. However, our construction is more general and flexible, as it can be viewed as the quantum analog of the construction given in \cite{Tamo}. More precisely, any partition of the set of evaluation points together with a good polynomial that is used to construct an LRC in \cite{Tamo} can also be used to construct a qLRC using our construction.


Next we now provide an example of our qLRC construction whose parameters  can not be obtained using the construction in \cite{Golowich}.  
\begin{example}
\rm{
Let $\alpha$ be a primitive element of $\mathbb{F}_{32},$ and define the additive subgroup $A_1 =\{0,1,\alpha,1+\alpha\}$ of $\mathbb{F}_{32}$ of order $4.$ Let $A_1, A_2,\ldots, A_8$ be the cosets of $A_1$ in $\mathbb{F}_{32},$ and let $A = \mathbb{F}_{32}$ be the evaluation set.
Then $g(x) = x^4 +(\alpha^2 +\alpha+1)x^2 +(\alpha^2 +\alpha)x$ is a good polynomial that is constant on the cosets $A_i$. One can easily verify that the vector $\boldsymbol{u}=(u_1,u_2, \ldots, u_{32}) = (1,1,\ldots,1)$ satisfies \eqref{tiki}. Further, set 
$$
S = \{x^ig(x)^j: 1\leq i\leq 2 ~\&~ 0\leq j\leq 4\} \cup \{xg(x)^5\}\cup\{g(x)^i:0\leq i\leq 7\}. 
$$ 
Then the code 
$$
\mathcal{C} = \left\{\ev_{A, \boldsymbol{u}}(f(x)) : f(x) \in \Span(S)\right\}
$$
is a dual-containing LRC over $\mathbb{F}_{32}$ of length $32$, dimension $19$, and locality $3$. 

The code $\mathscr{C} = \text{CSS}(\mathcal{C}, \mathcal{C})$ is a qLRC of length $32$, dimension $6$, alphabet size $32$, and locality $3$. The largest degree in $S$ is $28$ and, therefore, the minimum distance $\delta \ge n-28=4$ (see \eqref{eq:bound_degree} for details). Theorem~\ref{main-prop} in Section~\ref{sec:distance_bound} gives a better lower bound on the minimum distance: $\delta \ge 5$. }
\end{example}
\section{Good Polynomials from AGL}\label{sec:AGL}
In this section, we introduce a new approach for constructing good polynomials. This approach generalizes the good polynomial constructions in \cite{Tamo}. We begin by providing some basic definitions from group theory which are needed to describe our approach.  
\subsection{Group Action}
Let $G$  be a group and $X$ be a nonempty set. A \textit{group action} of $G$ on $X$ is a function $\bm{\cdot}: G \times X \to X$ that satisfies the following properties: 
\begin{enumerate}
    \item $e \bm{\cdot} x = x$ for all $x \in X$, where $e$ is the identity element of $G.$
    \item  $(gh) \bm{\cdot} x = g \bm{\cdot} (h \bm{\cdot} x)$ for all  $g, h \in G \text{ and } x \in X.$
\end{enumerate}
In the sequel, we will view the elements of $G$ as functions on $X$ and write $h(x)$ instead of $h\bm{\cdot} x$ for $h\in G$ and $x\in X$. 

We will also need the following definition. 
\begin{defi}
A group $G$ acts \emph{regularly} on a set $X$ if the action of $G$ on $X$ is both transitive and free. This means:
\begin{enumerate}
    \item (\emph{Transitivity}) For any $x,y \in X$, there exists $g \in G$ such that $g(x) = y $.
    \item (\emph{Free Action}) If $g(x) = x$ for some $x \in X$, then $g$ must be the identity element $e \in G$.
\end{enumerate}
\end{defi}
One can easily verify that in this case of regular action, the order of $G$ equals the size of $X$, i.e., $|G| = |X|$.

For an element $x\in X$, the \textit{orbit} of $x$ under the action of $G$ is defined as $$\mathrm{Orb}(x) = \{ g(x) \mid g \in G \}.$$ The orbit of $x\in X$ consists of all elements of $X$ that can be reached from $x$ by the action of the group $G.$

Note that 
\begin{enumerate}
     \item If $x, y \in X$ are two elements, then their orbits are either disjoint or identical, i.e., $$\mathrm{Orb}(x) \cap \mathrm{Orb}(y) = \emptyset \quad \text{if}\quad \mathrm{Orb}(x) \neq \mathrm{Orb}(y).$$
    \item The union of all orbits covers the entire set $X$, i.e.,  $$X = \bigcup_{x \in X} \mathrm{Orb}(x)$$
\end{enumerate}
Thus, the orbits form a partition of $X.$
\subsection{Affine General Linear Group}
We begin with the definition of the affine general linear group. 
 \begin{defi}
     Given a finite field $\mathbb{F}_q$, the affine general linear group $\text{AGL}$\footnote{
The more general definition of $\text{AGL}$ is a group that acts on the  $m$-dimensional vector space $\mathbb{F}_q^m$ for some $m.$ The definition above is restricted to the simple case where $m=1$.
} consists of all the transformations (polynomials) of the form $$ax + b\in \mathbb{F}_q[x], \text{ where } a,b \in \mathbb{F}_q \text{ and } a\neq 0,
$$
and where the group operation is function composition.
 \end{defi}
\noindent
Note that $\text{AGL}$ acts on the finite field $\mathbb{F}_q$ in the obvious way, i.e., for $\alpha \in \mathbb{F}_q$ and $f(x) = ax + b \in \text{AGL}$, we have 
$f(\alpha) = a\alpha + b.$

 The following theorem provides a method to obtain good polynomials using subgroups of $\text{AGL}$.  
\begin{theorem}\label{th4}
Let $H $ be a subgroup of $\text{AGL}.$ Then, the polynomial
\begin{equation*}\label{good}
    g(x) = \prod_{f \in H}(x - f(\alpha)) 
\end{equation*} for some $\alpha \in \mathbb{F}_q$, is constant on the orbits of $H.$ Moreover, the degree of the polynomial $g(x)$ is equal to the order of the group $H.$
\end{theorem}
\begin{proof}
Let $\beta$ and $\theta(\beta)$ be two elements in the same orbit, where $\theta(x) = ax + b \in H.$ It is well-known that the order of $\theta$, denoted by $\ord(\theta)$, is equal to the multiplicative order of $a$, i.e., $\ord(\theta) = \ord(a)$ and $\ord(\theta)$ divides the order of the group $H.$

Now, we observe that
\begin{align*}
    g(\theta(\beta)) &= \prod_{f \in H}(\theta(\beta) - f(\alpha)) = \prod_{f \in H}(a\beta + b - f(\alpha))\\ &= a^{|H|} \prod_{f \in H} \left(\beta - a^{-1}\left(f(\alpha) - b\right)\right).
\end{align*}
Since  $\theta^{-1} (x)=a^{-1}(x-b)\in H$ and the product runs over all $f\in H$, we have 
$$g(\theta(\beta)) = \prod_{f \in H} \left( \beta - \theta^{-1}(f(\alpha)) \right) = g(\beta).$$
Thus, $g(x)$ is constant on the orbits of $G$. The claim about the degree of $g(x)$ follows from its definition.
\end{proof}

\subsection{Good polynomials  via subgroups of AGL}  Let $H$ be a subgroup of $\text{AGL}$ of order $r+1$, such that it has $n/(r+1)$ orbits $A_i$ in $\mathbb{F}_q$, each of size $r+1$. Equivalently, $H$ acts regularly on each $A_i$.  Let $A = \sqcup_{i=1}^{\frac{n}{r+1}} A_i$ be the evaluation set of the code, consisting of $n$ elements. Then, by Theorem \ref{th4}, for some $\alpha \in \mathbb{F}_q$, the polynomial 
$$
g(x) = \prod\limits_{f \in H} (x - f(\alpha))
$$
is constant on each orbit $A_i$. Furthermore, since $\deg(g(x)) = r+1$, it follows that $g(x)$ is a good polynomial for the partition $\{A_1,A_2, \ldots, A_{\frac{n}{r+1}}\}$ of the set $A$.

In the above construction of good polynomials, we require $\frac{n}{r+1}$ orbits, each of size $r+1$, matching the order of the subgroup of $\text{AGL}$. Next, we present a method to achieve this.

Let $M, B \subseteq \mathbb{F}_q$. We say that $M$ is \textit{closed under multiplication by $B$} if, for all $a \in M$ and $b \in B$, the product $ab$ lies in $M$, i.e., $\{ab \mid a \in M, b \in B\} \subseteq M$.

\begin{lemma}\label{lem:mult_add}
Let $\mathbb{K}$ be a subfield of $\mathbb{F}_q$, so that $\mathbb{F}_q$ is a vector space over $\mathbb{K}$. Let $B \subseteq \mathbb{F}_q$ be a subspace over $\mathbb{K}$, and thus closed under multiplication by $\mathbb{K}$. Finally, let $M \subseteq \mathbb{K}^\ast$ be a multiplicative subgroup of $\mathbb{K}$. Then, 
$$
H = \{ ax + b \mid a \in M \text{ and } b \in B \} 
$$ 
is a subgroup of $\text{AGL}$ with order $|H| = |M| \cdot |B|$. Furthermore, for any $\alpha \in \mathbb{F}_q$, the size of its orbit under the action of $H$ is given by
$$
|\mathrm{Orb}(\alpha)| =
\begin{cases}
    |B| & \text{if } \alpha \in B, \\
    |H| & \text{if } \alpha \notin B.
\end{cases}
$$
\end{lemma}

See Appendix~\ref{appendix:mult_add} for the proof of the lemma. 

Lemma \ref{lem:mult_add} gives $(q - |B|)/|H|$ orbits of size equal to $|H|$, meaning that $H$ acts regularly on these orbits. Therefore, if we let $H$ be a subgroup of $\text{AGL}$ as in Lemma~\ref{lem:mult_add}, then together with Theorem \ref{th4}, one can construct good polynomials that are constant on relatively many orbits. 

Next, we show some applications of this approach to construct good polynomials. We note that all these examples are already known and are given in Proposition 3.2 and Theorem 3.3 of \cite{Tamo}, although they use different terminology.
\begin{enumerate}
    \item Let $\mathbb{K} = \mathbb{F}_q$, $B = \{0\}$, and let $M$ be a multiplicative subgroup of $\mathbb{F}_q^\ast$. Then $H = \{ax : a \in M\}$, which is isomorphic as a group to $M$, via the group isomorphism $ax \mapsto a$. In particular, if we choose $\alpha = 1$ in Theorem \ref{th4}, we get  
    $$
    g(x) = \prod_{f \in H}(x - f(1)) = \prod_{a \in M}(x - a) = x^{|M|} - 1,
    $$
    which is a good polynomial and coincides with the polynomial $g(x)$ stated in Proposition 3.2 of \cite{Tamo}. Similarly, if we pick $\alpha = 0$, we get the good polynomial $g(x) = x^{|M|}$.

    \item Let $\mathbb{K}$ be the prime subfield of $\mathbb{F}_q$. Set $M = \{1\}$ and let $B$ be an additive subgroup of $\mathbb{F}_q$. Then we have the subgroup $H = \{x + b : b \in B\}$, which is isomorphic to $B$, via the group isomorphism $x + b \mapsto b$. In particular, if we choose $\alpha = 0$ in Theorem \ref{th4}, we get  
    $$
    g(x) = \prod_{f \in H}(x - f(0)) = \prod_{b \in B}(x - b),
    $$
    which is a good polynomial and coincides with the polynomial $g(x)$ in Proposition 3.2 of \cite{Tamo}.

    \item If both $M$ and $B$ are non-trivial multiplicative and additive subgroups, respectively, i.e., $M \neq \{1\}$ and $B \neq \{0\}$, then $H$ is neither isomorphic to $M$ nor to $B$. In this case, we have $H = \{ax + b \mid a \in M \text{ and } b \in B\}$. In particular, if we choose $\alpha = 1$ in Theorem \ref{th4}, we get  
    $$
    g(x) = \prod_{f \in H}(x - f(1)) = \prod_{a \in M}\prod_{b \in B}(x - a - b),
    $$
    which is a good polynomial and coincides with the polynomial $g(x)$ in Theorem 3.3 of \cite{Tamo}.
\end{enumerate}

 Apart from the good polynomials in \cite{Tamo}, there are other constructions of good polynomials known in the literature, for instance, those in \cite{LiuMC18,Micheli, LiuMT20, RuiKai}. We leave it as an open question whether these good polynomials could also be constructed using the above $\text{AGL}$ approach. Furthermore, it remains an open question to characterize all good polynomials given by this approach, and to understand whether it provides new constructions of good polynomials with new parameters.
\section{Distance Bounds}\label{sec:distance_bound}
 In this section, we establish lower bounds on the minimum distance of qLRCs constructed in Section~\ref{sec:constrcution}.  
We first establish a simple lower bound based on the degree of evaluation polynomials.

For the LRC given by Construction~\ref{constr}, the minimum distance satisfies $d \ge n- \max_{f \in \Span(S)} \deg(f)$.  The maximum degree of polynomials in $\Span(S)$ is upper bounded by $\max\{n-(r+1), \ell\}$, where $\ell$ is defined in \eqref{eq:largest_deg}. Plugging this, we get that 
$$d \ge \min\{r+1, n-\ell\}.$$ 
Then, by Corollary~\ref{cor1}, we get the following lower bound. 
\begin{cor}
The minimum distance $\delta$ of the qLRC constructed using Construction~\ref{constr}  satisfies 
\begin{equation} \label{eq:bound_degree}
    \delta  \ge \max\{r+1, n-\ell\}.
\end{equation}
    
\end{cor}
 In the remainder of this section, we establish a new lower bound on the minimum distance of qLRCs constructed by good polynomials arising from subgroups of AGL, as described in Section \ref{sec:AGL}. This lower bound improves upon the bound in \eqref{eq:bound_degree} for most parameters, though not for all. In \cite{Golowich}, a lower bound on the minimum distance is derived for their qLRC, under the restriction that $r+1$ is prime. Our proof generalizes their result in two significant ways. 

First, we remove the restriction that $r+1$ must be prime, allowing it to hold for any positive integer. Additionally, our result applies to qLRCs constructed by any subgroup of AGL, whereas their construction is limited to a very specific family of subgroups of AGL, although they did not explicitly use this terminology.

To obtain our result, we combine the proof technique from \cite{Golowich} with some results from graph theory, particularly the Expander Mixing Lemma.

We begin by introducing the necessary graph-theoretic tools.

\begin{defi}[Schreier graph]
Let $G$ be a finite group that acts on a set $X$, and let $S$ be a symmetric subset of $G$, i.e., $S=S^{-1}$ such that for all distinct $s, s' \in S$ and $x \in X$,  $sx \neq s'x$. Then, the Schreier graph $Sch(X, S)$ is defined to be the graph with vertex set $X$ and edge set $\{ \{sx, x\} : x \in X, s \in S \}$.
\end{defi}
\noindent
Note that the above definition of a Schreier graph allows self-loops; however, it does not permit multiple edges, i.e., this is not a multigraph.
Furthermore,  a Schreier graph is a generalization of the more familiar family of graphs known as Cayley graphs, where $X = G$, i.e., the group acts on itself using the obvious action by left multiplication.


Recall that a $d$-regular graph is a graph in which every vertex is connected to exactly $d$ other vertices. We will need the following lemma, known as the Expander Mixing Lemma.
\begin{prop}[Expander Mixing Lemma \cite{Alon}]\label{Expander} 
Let $\mathcal{G} = (V, E)$ be a $d$-regular graph with $n$ vertices, and let $\lambda$ be the second largest eigenvalue in the absolute value of $\mathcal{G}$. For $S, T \subseteq V$, let $e(S, T)$ denote the number of edges between $S$ and $T$. Then
$$
\left| e(S, T) - \frac{d|S||T|}{n} \right| \leq \lambda \sqrt{|S||T| \left( 1 - \frac{|S|}{n} \right) \left( 1 - \frac{|T|}{n} \right)}.
$$
\end{prop}

 The  following theorem is the main result of this section, and it provides a lower bound on the minimum distance  of the qLRC constructed using a subgroup of AGL.
\begin{theorem}
\label{main-prop}
 Let   $\mathcal{C}$ be the $[n,k]_q$ code with locality $r$ obtained by Construction~\ref{constr}, constructed using a good polynomial that arise from a subgroup $H$   of $\text{AGL}$ as is Lemma \ref{lem:mult_add}. Then,
  the corresponding qLRC $\text{CSS}(\mathcal{C}, \mathcal{C})$ has minimum distance
  \begin{equation}\label{eq:distance}\delta \ge 
    n\left(1-\frac{1}{2p}-\sqrt{\frac{1}{4p^2}+\frac{p-1}{p}\cdot\frac{\ell-1}{n}}\right),
    \end{equation}
    where $p$ is the smallest prime factor of $r+1.$   
    \end{theorem}
\begin{proof}
Let $\mathtt{f}=\ev(f(x))\in \mathcal{C}\setminus \mathcal{C}^\perp$ be an arbitrary nonzero codeword. We aim to show that $\text{wt}(\mathtt{f})\geq \delta.$

Since $\mathcal{C}=\{\ev(f(x)): f(x)\in\Span(S)\},$ where $S=S_1\cup S_2$ such that $S_1$ and $S_2$ are disjoint, 
we can express $f(x)$ as
$$
f(x)=\gamma(x)+m(x),
$$
where $\gamma(x)\in\Span(S_1)$ and $m(x)\in\Span(S_2).$
Note that necessarily, $\gamma(x) \neq 0$, as otherwise $\ev(f) = \ev(m) \in \mathcal{C}^\perp$ since $S_2 \subseteq T$ and $\mathcal{C}^\perp = \{\ev(b): b(x) \in \Span(T)\}$. Furthermore, note that $m(x)$ is constant on the orbits of $H$.  

Define
\begin{equation}\label{eq:polynomial_G}
G(x)= \prod\limits_{t\in H\setminus\Theta} \frac{\gamma(t(x))-\gamma(x)}{t(x)-x},
\end{equation}
where
\begin{equation*}\label{eq:Theta}
    \Theta = \{ t(x) \in H \mid \gamma(t(x)) = \gamma(x) \}.
\end{equation*}
We will need the following lemma on the set $\Theta$. 
\begin{lemma}\label{lem4} 
$\Theta$ is a proper subgroup of $H$, and therefore, 
$H\setminus \Theta$ is closed under inverses, i.e., symmetric.
\end{lemma}
By Lemma \ref{lem4}, since  $\Theta$ is a proper subgroup of $H$, $G(x)$ is well-defined.
Moreover, since $\gamma(t(x)) - \gamma(x)$ is divisible by $t(x)-x$, we conclude that $G(x)$ is a nonzero polynomial whose degree satisfies
$$
\deg(G(x)) = |H\setminus\Theta|(\deg(\gamma(x))-1) \leq \mu(\ell-1),
$$
where $\mu = |H\setminus\Theta|$.

Next, we bound the number of roots of $G(x)$ in terms of the number of roots of $f(x)$. 
If, for some evaluation point $\alpha \in A$, both $f(\alpha) = 0$ and $f(\tilde{t}(\alpha)) = 0$ for some $\tilde{t} \in H \setminus \Theta$, then
\begin{align*}\label{sta}
    \gamma(\tilde{t}(\alpha)) - \gamma(\alpha) \notag &= -m(\tilde{t}(\alpha)) + m(\alpha)\\
    &= -m(\alpha) + m(\alpha) = 0,
\end{align*}
where the second equality follows since $m(x)$ is constant on the orbits of $H$. 
Therefore, we conclude that $G(\alpha) = 0$. Similarly, $G(x)$ has a root also at $\tilde{t}(\alpha)$ from the product in \eqref{eq:polynomial_G} with $\tilde{t}=t^{-1}$. Hence, any pair of such roots contributes exactly two roots to $G(x)$.

For the orbit $A_i, i = 1,2, \ldots, \frac{n}{r+1}$, we would like to count the number of ordered pairs $(\alpha, \beta)$, where $\alpha, \beta \in A_i$ such that $f(\alpha) = f(\beta) = 0$ and $t(\alpha) = \beta$ for some $t \in H \setminus \Theta$. For this, we consider the Schreier graph $Sch(A_i, H \setminus \Theta)$ defined on the orbit $A_i$ of $H$ and the set  $H \setminus \Theta$. Note that it is indeed a Schreier graph since 
(i) $t(\alpha) \neq \tilde{t}(\alpha)$ and  for any distinct $t, \tilde{t} \in H\setminus \Theta$ and $\alpha \in A_i$, as $H$ acts regularly on the orbit $A_i$ and (ii) the set $H\setminus \Theta$ is symmetric.

The largest eigenvalue of the graph equals $\mu = |H \setminus \Theta|$, and let $\lambda$ be the second largest eigenvalue in absolute value. Finally, let $S_i = \{\alpha \in A_i : f(\alpha) = 0\}$ be the set of zeros of $f$ in $A_i$. 

By Proposition \ref{Expander}, we have
$$
\left| e(S_i, S_i) - \frac{\mu |S_i|^2}{r+1} \right| \leq \lambda |S_i| \left( 1 - \frac{|S_i|}{r+1} \right),
$$
where $e(S_i, S_i)$ denotes twice the number of edges in the induced subgraph on the set $S_i$. Note that each edge is counted twice, as needed, since we would like to  count  ordered pairs and not only edges.
From this, we obtain
\begin{equation}\label{eq:the-lower-bound}   
e(S_i, S_i) \geq \frac{\mu |S_i|^2}{r+1} - \lambda |S_i| \left( 1 - \frac{|S_i|}{r+1} \right) 
= \frac{\mu + \lambda}{r+1} |S_i|^2 - \lambda |S_i|.
\end{equation}
Thus, $G(x)$ has $e(S_i, S_i)$ roots within a given orbit $A_i$ and the total number of roots of $G(x)$ is at least the summation of $e(S_i, S_i)$ over all the orbits $A_i$.

We will need the following lemma. 
\begin{lemma}
\label{good-lemma}
   $\lambda$ the second largest absolute value of the eigenvalues of  $Sch(A_i,H\setminus \Theta)$ equals $|\Theta|$ and therefore, $\mu+\lambda=|H|=r+1$.
\end{lemma}

Note that $|S_i| = (r+1 - \text{wt}(\mathtt{f}_i))$, where $\text{wt}(\mathtt{f}_i)$ denotes the Hamming weight of the restriction of $\mathtt{f}$ to $A_i$. Therefore, by Lemma \ref{good-lemma} and \eqref{eq:the-lower-bound}, the total number of roots of $G(x)$ is at least
\begin{align*}
     \sum_{i=1}^{\frac{n}{r+1}} |S_i|^2 - |\Theta| \sum_{i=1}^{\frac{n}{r+1}} |S_i|
    &=  \sum_{i=1}^{\frac{n}{r+1}} (r+1 - \text{wt}(\mathtt{f}_i))^2 - |\Theta| \sum_{i=1}^{\frac{n}{r+1}} (r+1 - \text{wt}(\mathtt{f}_i)) \\
    &\geq \frac{r+1}{n}(n - \text{wt}(\mathtt{f}))^2 - |\Theta| (n - \text{wt}(\mathtt{f})) \\
    &= \frac{r+1}{n} \text{wt}(\mathtt{f})^2 - (2(r+1)-|\Theta|) \text{wt}(\mathtt{f}) + n \mu.
\end{align*}
Since $G(x)$ is a nonzero polynomial of degree at most $ \mu(\ell-1)$, we must have
$$
\frac{r+1}{n} \text{wt}(\mathtt{f})^2 - (2(r+1)-|\Theta|) \text{wt}(\mathtt{f}) + \mu(n - \ell + 1) \leq 0.
$$
Solving this quadratic inequality and rearranging terms gives that 
\begin{equation}\label{eq:last-eq}
    \text{wt}(\mathtt{f}) \ge n\left(1-\frac{|\Theta|}{2(r+1)}-\sqrt{\frac{|\Theta|^2}{4(r+1)^2}+\frac{r+1-|\Theta|}{r+1}\cdot\frac{\ell-1}{n}}\right).
    \end{equation} Note that the right hand side of \eqref{eq:last-eq} is a decreasing function of $|\Theta|$, and recalling that $\Theta$ is a proper subgroup of $H$, and therefore its size is at most $(r+1)/p,$ where $p$ is the smallest prime factor of $r+1.$ Plugging this in \eqref{eq:last-eq} gives the desired bound in  \eqref{eq:distance}.
 \end{proof}
We are left to prove Lemmas \ref{lem4} and \ref{good-lemma}. 

\vspace{0.2cm}
\noindent{\bf Proof of Lemma \ref{lem4}:}
Since $H$ is finite, in order to show that $\Theta$ is a subgroup it is sufficient to show that $\Theta \neq \emptyset$ and that it is closed under composition. Indeed, the identity map $e(x) = x$ is in $\Theta$ as $\gamma(e(x)) = \gamma(x)$. 

Let $f_1, f_2 \in \Theta$. Since $\gamma(f_1(x)) = \gamma(x)$ and $\gamma(f_2(x)) = \gamma(x)$, we have
\[
\gamma((f_1 \circ f_2)(x)) = \gamma(f_1(f_2(x))) = \gamma(f_2(x)) = \gamma(x),
\]
which shows that $f_1 \circ f_2 \in \Theta$.

Next, since $\Theta$ is a subgroup, it is symmetric, and therefore, also $H\setminus \Theta$ is symmetric. 
Lastly, $\Theta$ is a proper subgroup, as otherwise, if $H= \Theta$, this would imply that $\gamma(x)$ is constant on the orbits of $H$, and therefore $\gamma(x)\in\Span(S_2)$. However, since $\gamma(x)\in\Span(S_1)$ as well, we would have $\gamma(x)\in\Span(S_2)\cap\Span(S_1)=\{0\}$, which is a contradiction.

\vspace{0.2cm}
 {\bf Proof of Lemma \ref{good-lemma}:}
Let $\mathcal{B}$ be the adjacency matrix of the graph $Sch(A_i, H \setminus \Theta)$, where the rows and columns are indexed by the elements of the orbit $A_i$. Similarly, let $\mathcal{B}'$ be the adjacency matrix of the Schreier graph $Sch(A_i, \Theta)$. Then, it is easy to see that 
$$
\mathcal{B} = J - \mathcal{B}',
$$ 
where $J$ is the square matrix of order $r+1$ with all entries equal to $1$. 

It is clear that the largest eigenvalue of $\mathcal{B}$ equals the degree, which is $r+1 - |\Theta|$, and the second largest eigenvalue in absolute value, denoted $\lambda$, satisfies 
$$
\lambda = \max_{\substack{\boldsymbol{v} \perp (1,1, \ldots, 1) \\ ||\boldsymbol{v}|| = 1}} |\boldsymbol{v}\mathcal{B}\boldsymbol{v}^t| = \max_{\substack{\boldsymbol{v} \perp (1,1, \ldots, 1) \\ ||\boldsymbol{v}|| = 1}} |\boldsymbol{v}\mathcal{B}'\boldsymbol{v}^t|.
$$ 

Since $H$ acts regularly on $A_i$, the subgroup $\Theta$ also acts regularly on each of its orbits in $A_i$. Specifically, $\Theta$ partitions $A_i$ into 
$$
\frac{|A_i|}{|\Theta|} = \frac{r+1}{|\Theta|}
$$
orbits, each of size $|\Theta|$. By indexing the rows and columns of $\mathcal{B}'$ according to these orbits, we observe that $\mathcal{B}'$ is a block diagonal matrix with $\frac{r+1}{|\Theta|}$ blocks, where each block is a square matrix of order $|\Theta|$ with all entries equal to $1$. Therefore, the largest eigenvalue of $\mathcal{B}'$ equals $|\Theta|$, or equivalently, 
$$
\max_{||\boldsymbol{v}||=1} |\boldsymbol{v}\mathcal{B}'\boldsymbol{v}^t| = |\Theta|.
$$ 

We will now show that, in fact, $\lambda = |\Theta|$. Indeed, let $\mathtt{e} \in \mathbb{C}$ be a primitive $\frac{r+1}{|\Theta|}$-th root of unity, and define the vector 
$$
\boldsymbol{v} = (\bold{1}_{|\Theta|}, \mathtt{e} \cdot \bold{1}_{|\Theta|}, \ldots, \mathtt{e}^{\frac{r+1}{|\Theta|}-1} \cdot \bold{1}_{|\Theta|}),
$$
where $\bold{1}_{|\Theta|}$ is the vector of all ones of length $|\Theta|$. One can verify that $\boldsymbol{v}$ is orthogonal to the all-ones vector, and 
$$
\frac{|\boldsymbol{v}\mathcal{B}'\boldsymbol{v}^t|}{||\boldsymbol{v}||^2} = |\Theta|.
$$ 
Thus, we conclude that $\lambda = |\Theta|$, and the result follows.
 
\vspace{0.2cm}
For the special case when $|H|=r+1$ is prime, we get the following corollary from Theorem \ref{main-prop}.
\begin{cor}\label{cor:prime_loc_bound}
Let $\mathcal{C}$ be the $[n, k]_q$ code with locality $r$ obtained by Construction~\ref{constr}, constructed using a good polynomial that arises from a subgroup $H$ of $\text{AGL}$ as in Lemma \ref{lem:mult_add}. Then, if $r + 1$ is prime, the corresponding qLRC $\text{CSS}(\mathcal{C}, \mathcal{C})$ has minimum distance \begin{equation}\label{eq:distance_prime}
    \delta \geq n\left(1 - \frac{1}{2(r+1)} - \sqrt{\frac{1}{4(r+1)^2} + \frac{r}{r+1} \cdot \frac{\ell-1}{n}}\right).
\end{equation}
\end{cor}

As mentioned earlier, in \cite{Golowich}, a qLRC construction with $n = q - 1$ and $(r + 1) \mid n$ is presented, along with a lower bound on the minimum distance of this construction for the case when $r + 1$ is prime. However, the bound given in \eqref{eq:distance_prime} is slightly better than that in \cite{Golowich}.

In Figure~\ref{fig:bound_plot}, we compare the distance bounds in \eqref{eq:bound_degree} and \eqref{eq:distance_prime} against the bound from \cite{Golowich} for an example parameter set. It can be seen that for fixed $n, r$, and different values of $\kappa$, the bound in \eqref{eq:distance_prime} is either one better than the bound in \cite{Golowich} or equal to it. This trend appears to hold for all possible parameters.

We also remark that, using graph-theoretic techniques similar to those used in this paper, one can obtain a lower bound on the minimum distance of the qLRC in \cite{Golowich} without any restriction on $r + 1$.

 \begin{figure}[ht!]
	\begin{center}
	\includegraphics[width=0.6\columnwidth]{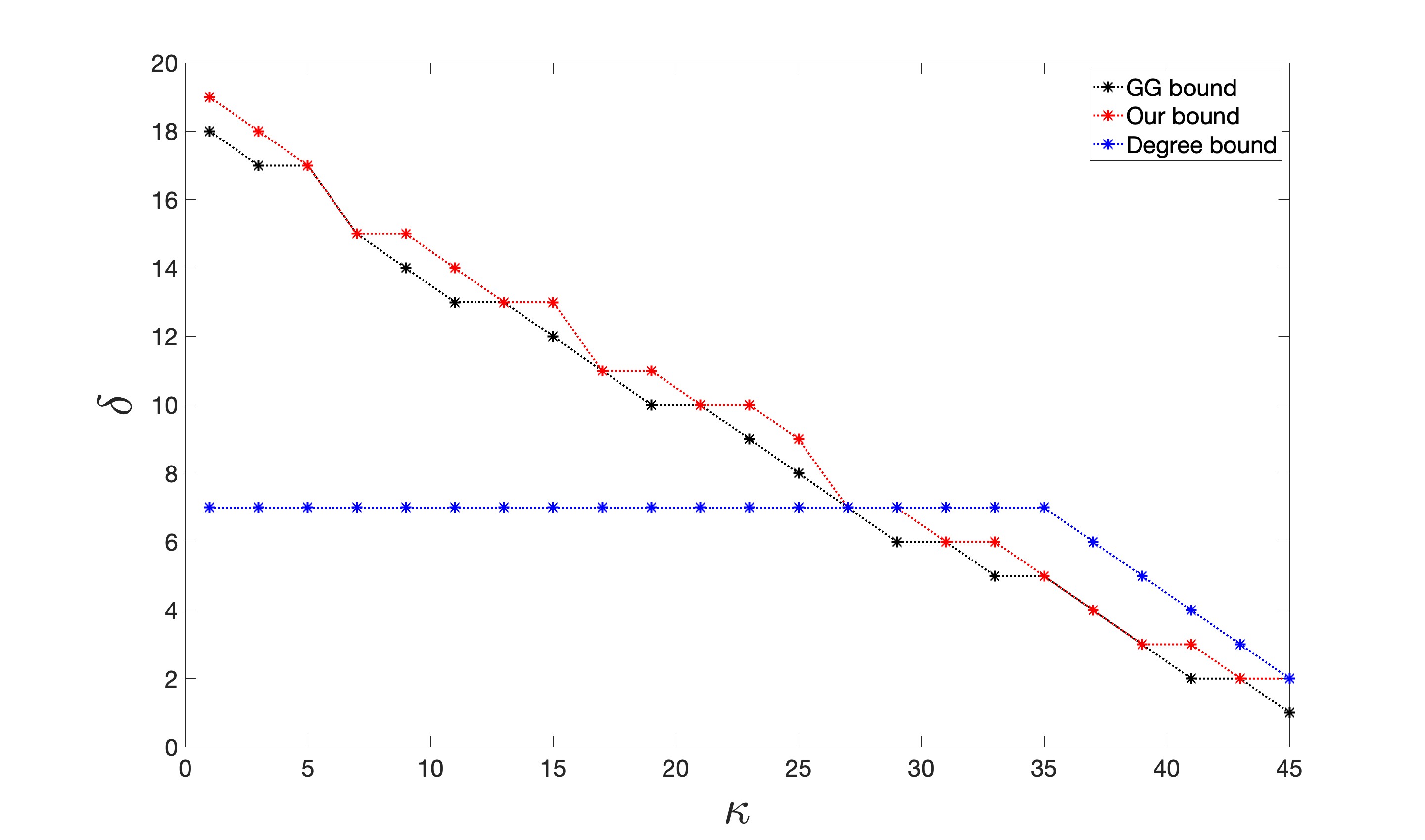} 
		\caption{Comparison between various bounds for $n=63, r=6, q=64$.
  Here GG bound is the bound in \cite{Golowich}, our bound is the bound in \eqref{eq:distance_prime}, and degree bound is the bound in \eqref{eq:bound_degree}.}
		\label{fig:bound_plot}
	\end{center} 
\end{figure}
The following result is another corollary that follows from Theorem \ref{main-prop} in a special case.
\begin{cor} \label{corr:additive_bound}
    Let $\mathcal{C}$ be an $[n, k]_q$ code with locality $r$ obtained by Construction \ref{constr} using a subgroup of $\text{AGL}$ that is isomorphic to an additive subgroup of $\mathbb{F}_q$, i.e., a subgroup of the form  
$\{x+b:b\in B\}$, where $B$ is an additive subgroup of $\mathbb{F}_q.$
  Then the corresponding qLRC $\text{CSS}(\mathcal{C}, \mathcal{C})$ has minimum distance

    \begin{equation*}\delta \ge 
    n\left(1-\frac{1}{2p}-\sqrt{\frac{1}{4p^2}+\frac{p-1}{p}\cdot\frac{\ell-1}{n}}\right),
    \end{equation*}
where $p$ is the characteristic of the field.       

\end{cor}

 



 \section{Conclusions and future work} \label{sec:conclusion}
We presented a construction of qLRCs based on the CSS approach by modifying the good polynomial-based LRC construction in \cite{Tamo} to make it dual-containing. A new framework for constructing good polynomials was introduced, and we derived a lower bound on the minimum distance of the resulting qLRCs. Further exploration is needed to determine if improved lower bounds can be derived for our construction.

In \cite{Golowich}, given an additional folding parameter \( s \mid (q - 1) / r \), a qLRC with alphabet size \( q \) and block length \( q - 1 \) is folded to obtain a qLRC with alphabet size \( q^s \) and block length \( (q - 1) / s \). The folded code retains the local recoverability of the unfolded code. It is possible to obtain a folded version of our qLRCs in a similar manner. A lower bound on the minimum distance of the folded code is derived in \cite{Golowich}, and this bound approaches the quantum Singleton-like bound \eqref{eq:quantum_singleton} for large values of \( r \). A future research direction would be to explore the possibility of deriving distance lower bounds for the folded version of our qLRCs that improve upon the bound in \cite{Golowich}.   

We note that our qLRC construction meets the quantum Singleton-like bound \eqref{eq:quantum_singleton} with equality if \( 2 \leq \delta \leq r + 2 + \frac{n + \kappa}{2} - \left\lceil \frac{n + \kappa}{2r} \right\rceil (r + 1) \); we omit the details. Constructing qLRCs with optimal distance for general parameters remains an open question.

\appendix
 

\subsection{Proof of Lemma~\ref{lem:good_poly}} \label{appendix:good_poly}


For $i = 1,2, \ldots, \frac{n}{r+1}-1$, let $f_i(x)$ be the minimal degree polynomial that satisfies 
$$
f_i(x) = \begin{cases}
1 & \text{if } x \in A_i, \\
0 & \text{if } x \notin A_i.
\end{cases}
$$
By Lagrange interpolation, $\deg(f_i) < n$, and therefore $f_i \in R$. Clearly, these $f_i$'s are linearly independent because each one is nonzero only on a specific subset $A_i$. We will see that, in fact, they form a basis for $R$. Indeed, any polynomial $f \in R$ can be written as 
$$
f(x) = \sum_{i=1}^{\frac{n}{r+1}-1} f(A_i) f_i(x),
$$
and thus $R$ has dimension $u$.

Next, let $w_1(x), w_2(x) \in R$. Note that both $w_1(x)$ and $w_2(x)$ are constant on each subset $A_i$, and $h(x)$ evaluates to zero on $A$. It follows that $w_1(x) w_2(x) \pmod{h(x)}$ is a polynomial that is constant on each subset $A_i$ and has degree less than $n$. Thus, $R$ is closed under multiplication modulo $h(x)$, which implies that $R$ is a ring under multiplication modulo $h(x)$.

Since $g(x)$ is a good polynomial for the set $A$, it is constant on each subset $A_i$  and has degree $r+1$, so $g(x) \in R$. Thus, the $\frac{n}{r+1}-1$ polynomials $1, g(x), g(x)^2, \ldots, g(x)^{\frac{n}{r+1}-1}$ are constant on each $A_i$ and have degree less than $n$, and therefore are members of $R$. Moreover, they are linearly independent because they have distinct degrees, which implies that they form a basis for $R$.

\subsection{Proof of Lemma~\ref{calc}}
\label{appendix-calc}
It is enough to consider the case where both $k - u$ and $n - k - u$ are not divisible by $r - 1$. Otherwise, if, for example, $k - u$ is divisible by $r - 1$, we still have 
$$
\ell = (r + 1) \left( \frac{k - u}{r - 1} \right) - 2 \leq (r + 1) \left\lfloor \frac{k - u}{r - 1} \right\rfloor + (k - u) \bmod (r - 1).
$$
A similar argument applies when $n - k - u$ is divisible by $r - 1$.

Thus, we assume that both $k - u$ and $n - k - u$ are not divisible by $r - 1$.
$$
\ell' + \ell = (n - k - u) \bmod (r - 1) + (r + 1) \left\lfloor \frac{n - k - u}{r - 1} \right\rfloor + (k - u) \bmod (r - 1) + (r + 1) \left\lfloor \frac{k - u}{r - 1} \right\rfloor.
$$
Let $s := (n - k - u) \bmod (r - 1)$ and $t := (k - u) \bmod (r - 1)$. Then the above expression simplifies to 
$$
s + \frac{(r + 1)(n - k - u - s)}{r - 1} + t + \frac{(r + 1)(k - u - t)}{r - 1} = n - 2u + \frac{2(n - 2u - t - s)}{r - 1}.
$$
Note that 
\begin{equation} \label{lili}
    s = (n - k - u) \bmod (r - 1) = \left(\frac{n r}{r + 1} - k\right) \bmod (r - 1),
\end{equation}
where the RHS of \eqref{lili} is congruent to $\left(\frac{n}{r + 1} - k\right) \bmod (r - 1)$, which equals $(u - k) \bmod (r - 1)$.
Also, recall that $t = (k - u) \bmod (r - 1)$, so $s + t = r - 1$. Plugging this in, we get 
$$
n - 2u + \frac{2(n - 2u - t - s)}{r - 1} = n - 2,
$$
and the result follows.
\subsection{Proof of Lemma~\ref{lem:mult_add}} 
\label{appendix:mult_add}
Since $H$ is finite and non-empty, it is sufficient to show that it is closed under composition. 
Let $f_1(x) = ax + b$ and $f_2(x) = cx + d$ be elements in $H$. Then,
$$
(f_1 \circ f_2)(x) = a(cx + d) + b = acx + ad + b.
$$
Since $ac \in M$ and $ad + b \in B$, we have $f_1 \circ f_2 \in H$, which confirms that $H$ is a subgroup of $\text{AGL}$.

The claim about the order of $H$ is straightforward, as for each $a \in M$, there are $|B|$ distinct transformations of the form $ax + b$ for $b \in B$. Since there are $|M|$ distinct choices for $a$, we conclude that $|H| = |M| \cdot |B|$.

Lastly, for any $\alpha \in \mathbb{F}_q$, the orbit of $\alpha$ under the action of $H$ is given by 
\begin{align*}
    \mathrm{Orb}(\alpha) &= \{f(\alpha) \mid f(x) \in H\} \\
    &= \{a\alpha + b \mid a \in M \text{ and } b \in B\} \\
    &= \bigcup_{a \in M} (a\alpha + B),
\end{align*}
which is a union of cosets of $B$. Now, let $a, c \in M$ be two distinct elements. Then $a\alpha + B = c\alpha + B$ if and only if $\alpha(a - c) \in B$. This holds if and only if $\alpha \in B$, since $a - c \in \mathbb{K}$ is nonzero and $B$ is closed under multiplication by $\mathbb{K}$. Therefore, the size of the orbit is 
$$
|\mathrm{Orb}(\alpha)| =
\begin{cases}
    |B| & \text{if } \alpha \in B, \\
    |H| & \text{if } \alpha \notin B.
\end{cases}
$$

\bibliographystyle{IEEEtranN}
\bibliography{ref}
\end{document}